\documentclass[twocolumn]{article}
\usepackage[T1]{fontenc}
\usepackage[latin9]{inputenc}
\usepackage{textcomp}
\usepackage{amsmath}
\usepackage{amsthm}
\usepackage{amssymb}
\usepackage{graphicx}
\usepackage[unicode=true,
 bookmarks=true,bookmarksnumbered=true,bookmarksopen=true,bookmarksopenlevel=1,
 breaklinks=false,pdfborder={0 0 0},pdfborderstyle={},backref=false,colorlinks=false]
 {hyperref}
\hypersetup{pdftitle={Your Title},
 pdfauthor={Your Name},
 pdfpagelayout=OneColumn, pdfnewwindow=true, pdfstartview=XYZ, plainpages=false}

\makeatletter

\newcommand{\lyxmathsym}[1]{\ifmmode\begingroup\def\b@ld{bold}
  \text{\ifx\math@version\b@ld\bfseries\fi#1}\endgroup\else#1\fi}

\numberwithin{equation}{section}
\numberwithin{figure}{section}
\theoremstyle{plain}
\newtheorem{thm}{\protect\theoremname}
\theoremstyle{definition}
\newtheorem{defn}[thm]{\protect\definitionname}
\theoremstyle{plain}
\newtheorem{lem}[thm]{\protect\lemmaname}
\theoremstyle{remark}
\newtheorem{rem}[thm]{\protect\remarkname}

\usepackage[caption=false,font=footnotesize]{subfig}

\makeatother

\providecommand{\definitionname}{Definition}
\providecommand{\lemmaname}{Lemma}
\providecommand{\remarkname}{Remark}
\providecommand{\theoremname}{Theorem}

\begin{document}
\title{$H_{\infty}$ Inverse Optimal Attitude Tracking on the Special Orthogonal
Group $SO(3)$}

\author{Farooq~Aslam and M. Farooq Haydar\thanks{Farooq Aslam and M. Farooq Haydar are with the Department of Aeronautics
and Astronautics, Institute of Space Technology, Islamabad, Pakistan
(e-mail: farooq.aslam87@gmail.com and farooq.haydar@ist.edu.pk).}}

\maketitle
\begin{abstract}
The problem of attitude tracking using rotation matrices is addressed
using an approach which combines inverse optimality and $\mathcal{L}_{2}$
disturbance attenuation. Conditions are provided which solve the
inverse optimal nonlinear $H_{\infty}$ control problem by minimizing
a meaningful cost function. The approach guarantees that the energy
gain from an exogenous disturbance to a specified error signal respects
a given upper bound. For numerical simulations, a simple problem setup
from literature is considered and results demonstrate competitive
performance.
\end{abstract}

\section{Introduction}

Rigid-body attitude control is an extensively studied control problem
with numerous applications in aircraft, spacecraft, robotics, and
marine systems. Different attitude parametrizations, or coordinates,
have been used to develop a wide array of attitude control methods.
Among these, the rotation matrix or direction cosine matrix (DCM),
an element of the Special Orthogonal Group $SO\left(3\right)$, is
the only attitude representation which is both globally defined and
unique \cite{chaturvedi2011rigid}. Other attitude representations
either contain singularities (e.g., Euler angles) due to which they
are not globally defined, or provide a non-unique attitude representation
(e.g., quaternions) where two different coordinates describe the same
attitude. In the case of quaternions, the resulting ambiguity in attitude
representation can lead to \emph{unwinding} \cite{chaturvedi2011rigid}.

Due to the limitations inherent to various attitude representations,
several research efforts have sought to develop attitude controllers
directly on the Special Orthogonal Group $SO(3)$. For the attitude
tracking problem on $SO\left(3\right)$, a control design which has
received significant attention is a simple PD-type controller, so
named as it contains $\textit{proportional}$ and $\textit{derivative}$-like
terms representing the attitude and angular velocity errors, respectively
(see \cite{chaturvedi2011rigid,Kang1995,Lee2010CDC} and the references
therein). This controller has been shown to be almost semi-globally
exponentially stabilizing \cite{Lee2010CDC,Bullo2005}. Moreover,
for sufficiently large controller gains, the region of exponential
convergence covers almost the entire state space. Given the simplicity
and popularity of this controller, its disturbance rejection properties
are of considerable interest.

The $\mathcal{L}_{2}$ disturbance attenuation framework, or nonlinear
$H_{\infty}$ control, provides powerful tools for studying the disturbance
rejection properties of feedback controllers (see \cite{isidori1994h,van1993nonlinear,krener1994necessary}).
In the case of state feedback, the approach facilitates the development
of controllers which solve the suboptimal nonlinear $H_{\infty}$
problem, thereby ensuring that the energy gain from exogenous inputs,
such as an external disturbance, to a specified error signal respects
a given upper bound. 

Several papers have considered the problem of $\mathcal{L}_{2}$ disturbance
attenuation in the context of attitude control. In \cite{Dalsmo1997},
a suboptimal nonlinear $H_{\infty}$ state feedback problem is formulated
on $TSO\left(3\right)$, the tangent space of $SO\left(3\right)$,
and addressed using a quaternion-PD controller. Quaternion-based $H_{\infty}$
control is also investigated in \cite{Ikeda2011,Show2003} for state
feedback controllers, and in \cite{Cavalcanti2016} for PD control
with delayed state measurements. In \cite{Damaren2014}, $H_{\infty}$
attitude control is achieved with a state feedback controller using
Modified Rodrigues Parameters (MRPs). For control laws defined on
$SO\left(3\right)$, the suboptimal nonlinear $H_{\infty}$ problem
is addressed in \cite{Kang1995} for attitude errors smaller than
$90\lyxmathsym{\textdegree}$.

A powerful approach to robust stabilization is obtained by combining
the $\mathcal{L}_{2}$ disturbance attenuation framework with the
inverse optimal control method \cite{Freeman1996,Krstic98}. In the
disturbance-free case, inverse optimal attitude control has been studied
in \cite{Krstic99} for the Cayley-Rodrigues parameters, and in \cite{Bharadwaj1998}
for exponential coordinates. For bounded disturbances, \cite{Luo2005}
uses the $H_{\infty}$ inverse optimal control method, described in
\cite{Krstic98}, to establish the attitude tracking and $\mathcal{L}_{2}$
disturbance attenuation properties of a quaternion-PD state feedback
controller. Building further on these ideas, \cite{ParkYonmook2013}
demonstrates the robustness of quaternion-based PD control to unmodeled
actuator dynamics.

In this paper, we apply the inverse optimal $\mathcal{L}_{2}$ disturbance
attenuation framework \cite{Krstic98} to the problem of attitude
tracking on $SO\left(3\right)$. In particular, we develop a state
feedback controller directly on $SO\left(3\right)$ such that it solves
the inverse optimal nonlinear $H_{\infty}$ problem almost globally
(in the sense of \cite{chaturvedi2011rigid}), thereby respecting
a given upper bound on the energy gain from the disturbance to a specified
error signal. The main contribution of this work is the provision
of $H_{\infty}$ guarantees for a PD-type control law with scalar
gains, at the cost of some mild constraints on the controller gains.
We also demonstrate that the disturbance rejection properties of control
laws (with scalar gains) synthesized using a common configuration
error function on $SO(3)$, namely the chordal metric, tend to deteriorate
for very large errors. This is due to the reduction in control effectiveness
which has been observed in chordal metric-based control laws for very
large errors, and which is known to cause arbitrarily slow convergence
for errors close to $180\lyxmathsym{\textdegree}$ \cite{chaturvedi2011rigid,Lee2011ACC}.

The rest of the paper is organized as follows. Essential background
on nonlinear $H_{\infty}$ control and an important result in $H_{\infty}$
inverse optimal control are summarized in Section \ref{sec:Preliminaries},
along with some remarks on the notation. The attitude control problem
using rotation matrices is reviewed in Section \ref{sec:Problem-Formulation},
while the main results of this paper are presented in Section \ref{sec:Robust-Attitude-Tracking}.
In Section \ref{sec:Simulation-Results}, a simple simulation setup
is adapted from \cite{Luo2005}, and used to demonstrate the effectiveness
of the proposed control law in tracking and disturbance attenuation.
Finally, concluding remarks are given in the last section.

\section{Preliminaries\label{sec:Preliminaries}}

Let us consider a general nonlinear system:
\begin{align}
\dot{x} & =f(x)+g_{1}(x)d+g_{2}(x)u, & z & =\begin{bmatrix}h(x)\\
\sqrt{r}u
\end{bmatrix},\label{eq:Dynamics_Prelim}
\end{align}
where $x$ is the state vector, $d$ is an exogenous disturbance,
$u$ is the control input, $z$ is the penalized performance output
signal, $h\left(x\right)$ is a positive state penalty, and $r$ is
a positive scalar. We assume that the functions $f(x)$, $g_{1}(x)$,
$g_{2}(x)$, and $h(x)$ are smooth, and the origin $x=0$ is an equilibrium
point of \eqref{eq:Dynamics_Prelim}, i.e., $f(0)=h(0)=0$. Also,
the disturbance $d$ belongs to the set of bounded-energy signals,
i.e., $\int_{0}^{T}\left|d(t)\right|^{2}dt<\infty$ for all finite
$T\geq0$.

\subsection{Nonlinear $H_{\infty}$ Control}
\begin{defn}[Nonlinear $H_{\infty}$ control or $\mathcal{L}_{2}$-disturbance
attenuation problem \cite{krener1994necessary}]
 The goal is to find a state feedback $u=k(x)$ such that $\forall$
$T\geq0$, $x_{0}$, and $d(t)$, the $\mathcal{L}_{2}$ gain from
the disturbance $d$ to the output signal $z(x,u)$ is less than or
equal to $\gamma$. More precisely, for a positive function $S(x)\geq0$,
\begin{equation}
\int_{0}^{T}\left|z\left(t\right)\right|^{2}dt\leq S(x_{0})+\gamma^{2}\int_{0}^{T}\left|d\left(t\right)\right|^{2}dt.\label{eq:L2gain_Hinf-Krener}
\end{equation}
\end{defn}

As has been well established in the literature on nonlinear $H_{\infty}$
control, the requirement \eqref{eq:L2gain_Hinf-Krener} on the $\mathcal{L}_{2}$
gain is closely related to the notion of dissipativity. In particular,
we seek a control law such that a smooth candidate storage function
$V(x)\geq0$, with $V(0)=0$, is dissipative with respect to a given
supply rate $U\left(d,z\right)$, i.e., the following condition is
satisfied:
\begin{equation}
\dot{V}\leq U.\label{eq:Dissipation_Prelim}
\end{equation}
Along trajectories of \eqref{eq:Dynamics_Prelim}, we have that:
\begin{equation}
\dot{V}=V_{x}f+\left(V_{x}g_{1}\right)d+\left(V_{x}g_{2}\right)u,\label{eq:Vdot_Dynamics_Prelim}
\end{equation}
where $V_{x}$ is a row vector of partial derivatives with respect
to the state $x$. Consider the supply rate

\[
U\left(d,z\right)=\frac{\gamma^{2}}{4}d^{\top}d-\frac{1}{4}z^{\top}z,
\]
where $\gamma$ is a positive scalar. Substitute the penalized output
signal $z$ and the supply rate in \eqref{eq:Dissipation_Prelim}:

\begin{equation}
\dot{V}+\frac{1}{4}h^{\top}h+\frac{1}{4}ru^{\top}u-\frac{\gamma^{2}}{4}d^{\top}d\leq0.\label{eq:DissipationInequality_Prelim}
\end{equation}
Then, using \eqref{eq:Vdot_Dynamics_Prelim}, the worst-case disturbance
can be found as:

\begin{equation}
\left(V_{x}g_{1}\right)^{\top}-\frac{\gamma^{2}}{2}d^{*}=0\implies d^{*}=\frac{2}{\gamma^{2}}\left(V_{x}g_{1}\right)^{\top},\label{eq:WorstCaseDist_Prelim}
\end{equation}
and the optimal state feedback is given by:

\begin{equation}
u(x)=-\frac{2}{r}\left(V_{x}g_{2}\right)^{\top}.\label{eq:Feedback_Hinf}
\end{equation}
Substituting the worst-case disturbance $d^{*}$ and the above state
feedback in \eqref{eq:DissipationInequality_Prelim}, we arrive at
the following expression for the dissipativity condition, known as
the Hamilton-Jacobi-Isaacs (HJI) partial differential inequality:

\begin{equation}
V_{x}f-\frac{1}{r}\left|V_{x}g_{2}\right|^{2}+\frac{1}{\gamma^{2}}\left|V_{x}g_{1}\right|^{2}+\frac{1}{4}h^{\top}h\leq0.\label{eq:HJI_Inequality_Prelim}
\end{equation}

\subsection{$H_{\infty}$ Inverse Optimal Control}

A key step in solving the $H_{\infty}$ control problem involves finding
a storage function such that the HJI inequality is satisfied. In situations
where finding a suitable storage function isn't easy, the dissipativity
requirement can be addressed by invoking the notion of inverse optimality.
In the $H_{\infty}$ inverse optimal method, the candidate storage
function $V(x)$, re-interpreted as a candidate Lyapunov function,
is used to obtain a constructive $H_{\infty}$ state-feedback control
law. In particular, the inverse optimal approach constructs an optimal
control problem, subject to the dynamics \eqref{eq:Dynamics_Prelim},
whose value function is determined by the candidate Lyapunov function
$V(x)$.

In the following, we state a simplified version of an important result
on $H_{\infty}$ inverse optimal control. The result is from \cite{Krstic98},
and is used in \cite{Luo2005} to obtain a quaternion-based nonlinear
$H_{\infty}$ inverse optimal attitude tracking control law. 
\begin{lem}
\cite{Krstic98}\label{lem:InverseOptimal_Prelim} Consider the nonlinear
system \eqref{eq:Dynamics_Prelim}, a candidate Lyapunov function
$V\left(x\right)$, and the auxiliary system
\begin{equation}
\dot{x}=f(x)+\frac{1}{\gamma^{2}}g_{1}(x)\left(V_{x}g_{1}\right)^{\top}+g_{2}(x)u.\label{eq:AuxiliarySystem_Prelim}
\end{equation}
Suppose that the control law
\begin{equation}
u=\alpha(x)=-\frac{1}{r}\left(V_{x}g_{2}\right)^{\top}\label{eq:Feedback_AuxSys}
\end{equation}
globally asymptotically stabilizes \eqref{eq:AuxiliarySystem_Prelim}
with respect to $V(x)$. Then the control law \eqref{eq:Feedback_Hinf}
solves the inverse optimal $H_{\infty}$ problem for \eqref{eq:Dynamics_Prelim}
by minimizing the cost functional
\begin{multline}
J(u)=\sup_{d\in D}\bigg(\lim_{t\rightarrow\infty}\bigg\{4V(x(t))\\
+\int_{0}^{t}\left[l(x)+ru^{\top}u-\gamma^{2}d^{\top}d\right]d\tau\bigg\}\bigg),\label{eq:Hinf_InvOpt_OCP}
\end{multline}
where $D$ is the set of locally bounded functions of the state, and
the state penalty is
\begin{equation}
l\left(x\right)=-4\left[V_{x}f-\frac{1}{r}\left|V_{x}g_{2}\right|^{2}+\frac{1}{\gamma^{2}}\left|V_{x}g_{1}\right|^{2}\right].\label{eq:Hinf_InvOpt_StatePenalty}
\end{equation}
Furthermore, the value function of \eqref{eq:Hinf_InvOpt_OCP} is
$4V\left(x\right)$, the optimal cost equals $J\left(u\right)=4V\left(x\left(0\right)\right)$,
and the worst-case disturbance is given by \eqref{eq:WorstCaseDist_Prelim}.
Lastly, the function $V\left(x\right)$ solves the following HJI equation:
\begin{equation}
V_{x}f-\frac{1}{r}\left|V_{x}g_{2}\right|^{2}+\frac{1}{\gamma^{2}}\left|V_{x}g_{1}\right|^{2}+\frac{l\left(x\right)}{4}=0,\label{eq:HJI_Equation_Prelim}
\end{equation}
and the achieved disturbance attenuation level is
\begin{equation}
\int_{0}^{\infty}\left[l\left(x\right)+ru^{\top}u\right]dt\leq\gamma^{2}\int_{0}^{\infty}\left\Vert d\right\Vert ^{2}dt.\label{eq:Disturbance_attenuation_level-gen}
\end{equation}
\end{lem}

\begin{rem}
We note that the $H_{\infty}$ inverse optimal state penalty $l\left(x\right)$
in \eqref{eq:Hinf_InvOpt_StatePenalty} equals $-4\dot{V}_{\text{aux}}$,
where
\[
\dot{V}_{\text{aux}}=V_{x}f+\frac{1}{\gamma^{2}}\left|V_{x}g_{1}\right|^{2}-\frac{1}{r}\left|V_{x}g_{2}\right|^{2}
\]
is the Lyapunov rate for the closed-loop auxiliary system \eqref{eq:AuxiliarySystem_Prelim}-\eqref{eq:Feedback_AuxSys}.
From the assumptions of the lemma, we know that $\dot{V}_{\text{aux}}<0$.
As a result, $l\left(x\right)$ is positive definite, and the performance
index in \eqref{eq:Hinf_InvOpt_OCP} is a meaningful cost, since it
effectively penalizes the state and the control for each $\left(x,u\right)$
with a positive penality. Moreover, in the HJI equation \eqref{eq:HJI_Equation_Prelim},
the state penalty $l\left(x\right)$ replaces the term $h^{\top}h$
in the HJI inequality \eqref{eq:HJI_Inequality_Prelim}. Thus, the
$H_{\infty}$ inverse optimal method boils down to stabilizing the
auxiliary system using the $\left(L_{g}V\right)$-type state-feedback
\eqref{eq:Feedback_AuxSys}. Note that the auxiliary term in \eqref{eq:AuxiliarySystem_Prelim}
is also an $\left(L_{g}V\right)$-type term expressing the contribution
of the worst-case disturbance to the HJI inequality \eqref{eq:HJI_Inequality_Prelim}.
\end{rem}

\subsection{Notation}

The Special Orthogonal group $SO(3)$ is the set of $3\times3$ orthogonal
matrices with determinant $1$, i.e.,
\begin{align*}
SO(3) & =\left\{ R\in\mathbb{R}^{3\times3}:R^{\top}R=I,\det\left(R\right)=1\right\} .
\end{align*}
The \textit{hat} map $\wedge:\mathbb{R}^{3}\rightarrow\mathfrak{so}(3)$
transforms a vector in $\mathbb{R}^{3}$ to a $3\times3$ skew-symmetric
matrix such that $\hat{x}y=x\times y$ for any $x,y\in\mathbb{R}^{3}$.
Sometimes, $\hat{x}$ is written as $\left(x\right)^{\times}$ for
clarity. In particular,
\[
\hat{x}=\left(x\right)^{\times}=\begin{bmatrix}0 & -x_{3} & x_{2}\\
x_{3} & 0 & -x_{1}\\
-x_{2} & x_{1} & 0
\end{bmatrix}.
\]
The inverse of the hat map is denoted by the \textit{vee} map $\vee:\mathfrak{so}(3)\rightarrow\mathbb{R}^{3}$.
Finally, we recall some useful properties of the hat map \cite{Lee2011ACC}:

\begin{align}
\hat{x}y & =x\times y=-y\times x=-\hat{y}x,\label{eq:Identity1}\\
x^{\top}\hat{x} & =0,\quad\hat{x}=-\hat{x}^{\top},\label{eq:Identity2}\\
\text{tr}\left[A\hat{x}\right] & =\frac{1}{2}\text{tr}\left[\hat{x}\left(A-A^{\top}\right)\right]=-x^{\top}\left(A-A^{\top}\right)^{\vee},\label{eq:Identity3}\\
\hat{x}A+A^{\top}\hat{x} & =\left(\left\{ \text{tr}[A]I-A\right\} x\right)^{\times},\label{eq:Identity4}\\
R\hat{x}R^{\top} & =\left(Rx\right)^{\times},\label{eq:Identity5}
\end{align}
for any $x,y\in\mathbb{R}^{3}$, $A\in\mathbb{R}^{3\times3}$, and
$R\in SO(3)$.

\section{Problem Formulation\label{sec:Problem-Formulation}}

The equations of motion for rigid-body rotation can be written as:
\begin{equation}
\begin{split}\dot{R} & =R\omega^{\times}\\
J\dot{\omega} & =-\omega^{\times}J\omega+\tau+d,
\end{split}
\label{eq:EoMs}
\end{equation}
where $R\in SO(3)$ is the rotation matrix, $\omega\in T_{I}SO(3)=\mathfrak{so}(3)$
is the angular velocity vector expressed in the body-fixed frame,
$J$ is the moment of inertia matrix, $\tau$ is the net external
torque, and $d$ is the disturbance torque. We are interested in tracking
a given reference attitude which obeys the kinematics

\begin{equation}
\dot{R}_{d}=R_{d}\omega_{d}^{\times},\label{eq:ReferenceKinematics}
\end{equation}
where the subscript in $R_{d}$ stands for 'desired' or reference. 

We define the right attitude error \cite[pg. 554]{Bullo2005} as

\begin{equation}
R_{e}:=R_{d}^{\top}R\in SO(3).\label{eq:AttitudeError}
\end{equation}
Using \eqref{eq:EoMs}-\eqref{eq:AttitudeError} and the identity
\eqref{eq:Identity5}, the error kinematics can be expressed as

\begin{align*}
\dot{R}_{e} & =R_{e}\left(\omega-R_{e}^{\top}\omega_{d}\right)^{\times}.
\end{align*}
Next, we define the right angular velocity error as \cite[pg. 555]{Bullo2005}:

\begin{equation}
\omega_{e}:=\omega-R_{e}^{\top}\omega_{d}.\label{eq:VelocityError}
\end{equation}
Consequently, the error kinematics can be expressed as

\begin{equation}
\dot{R}_{e}=R_{e}\omega_{e}^{\times}.\label{eq:ErrorKinematics_Right}
\end{equation}
Using \eqref{eq:EoMs} and \eqref{eq:VelocityError}-\eqref{eq:ErrorKinematics_Right},
the error dynamics is given by:

\begin{align*}
J\dot{\omega}_{e} & =J\dot{\omega}-JR_{e}^{\top}\dot{\omega}_{d}-J\dot{R}_{e}^{\top}\omega_{d}\\
 & =-\left(\omega_{e}+R_{e}^{\top}\omega_{d}\right)^{\times}J\left(\omega_{e}+R_{e}^{\top}\omega_{d}\right)+\tau+d\\
 & \begin{aligned}\quad\; & -JR_{e}^{\top}\dot{\omega}_{d}+J\omega_{e}^{\times}R_{e}^{\top}\omega_{d}\end{aligned}
\end{align*}
Using identities \eqref{eq:Identity1} and \eqref{eq:Identity4}
and some algebraic manipulation, three of the resulting terms can
be simplified as follows:
\[
-\omega_{e}^{\times}JR_{e}^{\top}\omega_{d}-\left(R_{e}^{\top}\omega_{d}\right)^{\times}J\omega_{e}+J\omega_{e}^{\times}R_{e}^{\top}\omega_{d}=-\omega_{e}^{\times}\bar{J}R_{e}^{\top}\omega_{d}
\]
where $\bar{J}:=2J-\text{tr}[J]I$. Consequently, the error dynamics
can be expressed as
\begin{equation}
J\dot{\omega}_{e}=-\omega_{e}^{\times}J\omega_{e}+\tau+d_{e},\label{eq:ErrorDynamics}
\end{equation}
where

\begin{equation}
d_{e}:=d-\omega_{e}^{\times}\bar{J}R_{e}^{\top}\omega_{d}-JR_{e}^{\top}\dot{\omega}_{d}-\left(R_{e}^{\top}\omega_{d}\right)^{\times}JR_{e}^{\top}\omega_{d}\label{eq:ExtendedDisturbance}
\end{equation}
includes the disturbance torque $d\left(t\right)$ as well as the
terms containing the reference signals $\omega_{d}(t)$ and its time
derivative $\dot{\omega}_{d}(t)$. Henceforth, the vector $d_{e}\left(t\right)$
is referred to as the $\textit{extended}$ disturbance.

The choice of configuration and velocity error functions is pivotal
for control design on manifolds. We select the following configuration
error function:

\begin{equation}
\Psi(R_{e}):=\frac{1}{2}\text{tr}\left[I-R_{e}\right].\label{eq:ConfigurationErrorFunction}
\end{equation}
From \eqref{eq:ErrorKinematics_Right}, the time derivative of $\Psi$
is given by

\begin{align*}
\dot{\Psi} & =-\frac{1}{2}\text{tr}\left[\dot{R}_{e}\right]=-\frac{1}{2}\text{tr}\left[R_{e}\omega_{e}^{\times}\right]=\frac{1}{2}\left(R_{e}-R_{e}^{\top}\right)^{\vee}\cdot\omega_{e},
\end{align*}
where the last equality follows from the identity \eqref{eq:Identity3}.
Define the configuration error vector as

\begin{equation}
e_{R}:=\frac{1}{2}\left(R_{e}-R_{e}^{\top}\right)^{\vee}.\label{eq:ConfigurationErrorVector}
\end{equation}
Then, the rate of change of the configuration error function can be
expressed as follows:

\begin{equation}
\dot{\Psi}=e_{R}\cdot\omega_{e}.\label{eq:ConfigurationErrorFunction_TimeDerivative}
\end{equation}

The time derivative of the configuration error vector can be obtained
as follows:

\begin{align}
\dot{e}_{R}^{\times} & =\frac{1}{2}\left(\dot{R}_{e}-\dot{R}_{e}^{\top}\right)\nonumber \\
 & =\frac{1}{2}\left(R_{e}\hat{\omega}_{e}+\hat{\omega}_{e}R_{e}^{\top}\right)\nonumber \\
 & =\frac{1}{2}\left[\left(\text{tr}[R_{e}^{\top}]I-R_{e}^{\top}\right)\omega_{e}\right]^{\times}\nonumber \\
\implies\dot{e}_{R} & =\frac{1}{2}\left[\left(\text{tr}[R_{e}^{\top}]I-R_{e}^{\top}\right)\omega_{e}\right]\label{eq:eR_dot}
\end{align}
In the above simplification, the third line follows from the identity
\eqref{eq:Identity4}. Suppose that:

\begin{equation}
E(R_{e}):=\text{tr}[R_{e}^{\top}]I-R_{e}^{\top}.\label{eq:E}
\end{equation}
Then, the time derivative of the configuration error vector can be
expressed as follows:
\begin{equation}
\dot{e}_{R}=\frac{1}{2}E(R_{e})\omega_{e}.\label{eq:ConfigurationErrorVector_TimeDerivative}
\end{equation}

Next, consider a control law of the form

\begin{equation}
\tau=u\left(e_{R},\omega_{e}\right)+u_{\text{FF}},\label{eq:PD+}
\end{equation}
where $u\left(e_{R},\omega_{e}\right)\in\mathbb{R}^{3}$ is a feedback
term, to be specified below, and

\begin{equation}
u_{\text{FF}}:=\omega_{e}^{\times}\bar{J}R_{e}^{\top}\omega_{d}+\left(R_{e}^{\top}\omega_{d}\right)^{\times}JR_{e}^{\top}\omega_{d}+JR_{e}^{\top}\dot{\omega}_{d}\label{eq:Feedforward}
\end{equation}
is an optional feedforward compensation term which can be used to
cancel the contribution of the terms in \eqref{eq:ExtendedDisturbance}
which contain the reference angular velocity $\omega_{d}\left(t\right)$
or its time derivative $\dot{\omega}_{d}(t)$ . Consequently, from
\eqref{eq:ErrorDynamics} and \eqref{eq:ConfigurationErrorVector_TimeDerivative}-\eqref{eq:PD+},
with $u_{\text{FF}}=0$, the tracking error system with extended disturbance
can be expressed as:

\begin{equation}
\begin{split}\dot{e}_{R} & =\frac{1}{2}E(R_{e})\omega_{e},\\
J\dot{\omega}_{e} & =-\omega_{e}^{\times}J\omega_{e}+u+d_{e}.
\end{split}
\label{eq:TrackingErrorSystem_ExtendedDisturbance}
\end{equation}

Our goal is to find a feedback control law $u\left(e_{R},\omega_{e}\right)\in\mathbb{R}^{3}$
which achieves $H_{\infty}$ inverse optimal attitude tracking problem
on $SO(3)$. In other words, we seek a stabilizing feedback control
law which tracks the given reference signal ($R_{d},\omega_{d}$)
while respecting a given upper bound on the energy gain from the extended
disturbance $d_{e}$ to a linear combination of suitable state and
input penalties, as in \eqref{eq:Disturbance_attenuation_level-gen}.

\section{Robust Attitude Tracking on $SO(3)$\label{sec:Robust-Attitude-Tracking}}

\subsection{Candidate Storage Function}

We consider the following candidate storage function:

\begin{equation}
V\left(R_{e},\omega_{e}\right):=\frac{a}{2}\omega_{e}\cdot J\omega_{e}+be_{R}\cdot J\omega_{e}+2c\Psi\left(R_{e}\right),\label{eq:CandidateStorageFunction_SO3}
\end{equation}
where $\omega_{e}$ is the angular velocity error \eqref{eq:VelocityError},
$\Psi$ is the configuration error function \eqref{eq:ConfigurationErrorFunction},
$e_{R}$ is its associated configuration error vector \eqref{eq:ConfigurationErrorVector},
and $a,b,c$ are positive scalars. We note that the configuration
error function is bounded as \cite{Lee2010CDC}:

\begin{equation}
\frac{1}{2}\left\Vert e_{R}\right\Vert ^{2}\leq\Psi\left(R_{e}\right)\leq\frac{1}{2-\psi}\left\Vert e_{R}\right\Vert ^{2},\label{eq:Psi_Bounds}
\end{equation}
where $\psi$ is a constant such that $0<\psi<2$, and the upper bound
holds when $\Psi\left(R_{e}\right)\leq\psi.$ Using the lower bound,
we obtain
\[
V\geq\frac{1}{2}\begin{bmatrix}e_{R}\\
\omega_{e}
\end{bmatrix}^{\top}\begin{bmatrix}2cI & bJ\\
bJ & aJ
\end{bmatrix}\begin{bmatrix}e_{R}\\
\omega_{e}
\end{bmatrix}.
\]
Consequently, a sufficient condition for $V$ to be positive definite
is given by

\begin{equation}
2acI>b^{2}J.\label{eq:V_PosDef}
\end{equation}
Using \eqref{eq:ConfigurationErrorFunction_TimeDerivative}, we
have that:

\begin{align*}
\dot{V} & =a\omega_{e}\cdot J\dot{\omega}_{e}+b\dot{e}_{R}\cdot J\omega_{e}+be_{R}\cdot J\dot{\omega}_{e}+2c\dot{\Psi},\\
 & =\left(a\omega_{e}+be_{R}\right)\cdot J\dot{\omega}_{e}+b\dot{e}_{R}\cdot J\omega_{e}+2ce_{R}\cdot\omega_{e}.
\end{align*}
Substituting the error dynamics from \eqref{eq:TrackingErrorSystem_ExtendedDisturbance}
and using the identity \eqref{eq:Identity2}, it follows that:

\begin{multline*}
\dot{V}=-be_{R}\cdot\omega_{e}^{\times}J\omega_{e}+b\dot{e}_{R}\cdot J\omega_{e}+2ce_{R}\cdot\omega_{e}\\
+\left(a\omega_{e}+be_{R}\right)\cdot\left(u+d_{e}\right).
\end{multline*}
We observe that:
\[
e_{R}\cdot\omega_{e}^{\times}J\omega_{e}=-\omega_{e}^{\times}e_{R}\cdot J\omega_{e}=e_{R}^{\times}\omega_{e}\cdot J\omega_{e}
\]
Therefore, from \eqref{eq:ConfigurationErrorVector} and \eqref{eq:eR_dot}-\eqref{eq:E},
it follows that:
\begin{align*}
\dot{e}_{R}-e_{R}^{\times}\omega_{e} & =\frac{1}{2}\left(\left(\text{tr}[R_{e}^{\top}]I-R_{e}^{\top}\right)\omega_{e}-\left(R_{e}-R_{e}^{\top}\right)\omega_{e}\right)\\
 & =\frac{1}{2}\left(\text{tr}[R_{e}]I-R_{e}\right)\omega_{e}\\
 & =\frac{1}{2}E^{\top}(R_{e})\omega_{e}
\end{align*}
Consequently, the rate of change of the storage function along trajectories
of the error system \eqref{eq:TrackingErrorSystem_ExtendedDisturbance}
can be expressed as:

\begin{multline}
\dot{V}=\frac{b}{2}J\omega_{e}\cdot E^{\top}(R_{e})\omega_{e}+\left(a\omega_{e}+be_{R}\right)\cdot\left(u+d_{e}\right)\\
+2ce_{R}\cdot\omega_{e}\label{eq:Vdot_Sys1_A}
\end{multline}

\subsection{Auxiliary System}

In this section, we use the storage function discussed above to study
the robustness properties of the tracking error system \eqref{eq:TrackingErrorSystem_ExtendedDisturbance}.
We employ the inverse optimal robust control approach, as summarized
in Section \ref{sec:Preliminaries}. In particular, we consider an
auxiliary system associated with \eqref{eq:TrackingErrorSystem_ExtendedDisturbance},
and establish its closed-loop stability under a PD-like feedback control
law. 

The auxiliary tracking error system is given as:

\begin{equation}
\begin{split}\dot{e}_{R} & =\frac{1}{2}E(R_{e})\omega_{e},\\
J\dot{\omega}_{e} & =-\omega_{e}^{\times}J\omega_{e}+u+\frac{1}{\gamma^{2}}\left(a\omega_{e}+be_{R}\right)
\end{split}
\label{eq:AuxiliarySystem}
\end{equation}
It is seen that except for replacing the extended disturbance $d_{e}$
by the expression $\frac{1}{\gamma^{2}}\left(\omega_{e}+be_{R}\right)$,
the auxiliary system is similar to the tracking error system \eqref{eq:TrackingErrorSystem_ExtendedDisturbance}.

We recall that due to a topological restriction, smooth control laws
cannot achieve global stability in problems involving rotational degrees
of freedom \cite{bhat2000topological}. Therefore, we invoke standard
notions of almost global asymptotic stability (AGAS) and almost semi-global
exponential stability (AsGES) \cite{chaturvedi2011rigid,Lee2015}
to prove the stability properties of the auxiliary system. 
\begin{thm}
\label{thm:Theorem_AuxiliarySystem}Consider the auxiliary tracking
error system \eqref{eq:AuxiliarySystem} and the control law

\begin{equation}
u=-\frac{1}{r}\left(a\omega_{e}+be_{R}\right).\label{eq:StateFeedback_A}
\end{equation}
Suppose the coefficients $a,b,c$ in the candidate Lyapunov function
\eqref{eq:CandidateStorageFunction_SO3} and the scalars $\gamma$,
$r$ satisfy the following conditions:
\begin{align}
\gamma^{2} & >r>0, & c & =ab\alpha, & 0<b\lambda_{j}<a^{2}\alpha,\label{eq:Conditions_AuxSys_A}
\end{align}
where $\lambda_{j}$ is the maximum eigenvalue of $J$, and

\begin{equation}
\alpha:=\frac{1}{r}-\frac{1}{\gamma^{2}}.\label{eq:alpha}
\end{equation}
Then, the zero equilibrium $\left(e_{R},\omega_{e}\right)=\left(0,0\right)$
of \eqref{eq:AuxiliarySystem} is almost globally asymptotically stable
and almost semi-globally exponentially stable, i.e., the set of initial
conditions for which exponential stability is guaranteed almost covers
$SO\left(3\right)\times\mathbb{R}^{3}$ when $\alpha$ is sufficiently
large:
\begin{align}
\Psi\left(R_{e}\left(0\right)\right) & \leq\psi<2\nonumber \\
\left\Vert \omega_{e}\left(0\right)\right\Vert ^{2} & \leq\frac{2b\alpha}{\lambda_{j}}\left(\psi-\Psi\left(R_{e}\left(0\right)\right)\right).\label{eq:InitialConditions_ExponentialStability}
\end{align}
\end{thm}

\begin{proof}
The proof uses several ideas discussed in \cite{Lee2015}, and is
given in Appendix \ref{sec:Appendix_A}.
\end{proof}

\subsection{$H_{\infty}$ Inverse Optimal Attitude Tracking}

In this section, we build on earlier results in order to obtain a
robust attitude tracking control law on $SO\left(3\right)$. In particular,
we consider the tracking error system \eqref{eq:TrackingErrorSystem_ExtendedDisturbance},
and propose a state feedback which ensures $H_{\infty}$ inverse optimal
attitude tracking with respect to the extended disturbance $d_{e}$
specified in \eqref{eq:ExtendedDisturbance}.

We now state the main result concerning $H_{\infty}$ inverse optimal
attitude tracking.
\begin{thm}
\label{thm:MainResult}Consider the tracking error system \eqref{eq:TrackingErrorSystem_ExtendedDisturbance},
with the bounded extended disturbance $d_{e}$ defined as in \eqref{eq:ExtendedDisturbance}.
Consider also the inverse optimal $H_{\infty}$ control in which the
objective is to minimize the cost functional
\begin{multline}
J_{a}(u)=\sup_{d\in D}\bigg(\lim_{t\rightarrow\infty}\bigg\{4V\left(R_{e},\omega_{e}\right)\\
+\int_{0}^{t}\left[l\left(R_{e},\omega_{e}\right)+ru^{\top}u-\gamma^{2}d_{e}^{\top}d_{e}\right]d\tau\bigg\}\bigg),\label{eq:CostFunctional}
\end{multline}
where $V\left(R_{e},\omega_{e}\right)$ is the candidate storage function
\eqref{eq:CandidateStorageFunction_SO3}, the state penalty function
\begin{equation}
l:=4a^{2}\alpha\left\Vert \omega_{e}\right\Vert ^{2}+4b^{2}\alpha\left\Vert e_{R}\right\Vert ^{2}-2bJ\omega_{e}\cdot E^{\top}\left(R_{e}\right)\omega_{e},\label{eq:StatePenalty_Sys2_B}
\end{equation}
 and $r>0$ is a scalar penalizing the control effort. Suppose that
the coefficients $\alpha$, $a$, $b$, and $c$ in $V$ and $l$
satisfy the conditions \eqref{eq:Conditions_AuxSys_A}-\eqref{eq:alpha}
given in Theorem \ref{thm:Theorem_AuxiliarySystem}. Then:
\end{thm}

\begin{enumerate}
\item $V$ and $l$ are nonnegative, and the state-feedback PD control 
\begin{equation}
u=-\frac{2}{r}\left(a\omega_{e}+be_{R}\right).\label{eq:StateFeedback_Main_A}
\end{equation}
 solves the inverse optimal $H_{\infty}$ control problem for the
closed-loop attitude tracking system \eqref{eq:TrackingErrorSystem_ExtendedDisturbance}
with respect to the performance index \eqref{eq:CostFunctional},
where the worst-case extended disturbance $d_{e}$ is 
\begin{equation}
d_{e}^{*}=\frac{2}{\gamma^{2}}\left(a\omega_{e}+be_{R}\right).\label{eq:WorstCaseExtendedDisturbance}
\end{equation}
\item The optimal cost is $J_{a}(u)=4V_{0}=4V(R_{e}(0),\omega_{e}(0))$,
and the following disturbance attenuation level 
\[
\int_{0}^{\infty}\left[l\left(R_{e},\omega_{e}\right)+ru^{\top}u\right]dt\leq\gamma^{2}\int_{0}^{\infty}\left\Vert d_{e}\right\Vert ^{2}dt,
\]
is achieved.
\end{enumerate}
\begin{proof}
This theorem is a consequence of Lemma \ref{lem:InverseOptimal_Prelim}
and Theorem \ref{thm:Theorem_AuxiliarySystem}. In particular, we
note the following correspondence between the terms in \eqref{eq:Vdot_Dynamics_Prelim}
and those in the Lyapunov rate \eqref{eq:Vdot_Sys1_A} along trajectories
of the tracking error system \eqref{eq:TrackingErrorSystem_ExtendedDisturbance}:
\begin{align*}
V_{x}f & =\frac{b}{2}J\omega_{e}\cdot E^{\top}(R_{e})\omega_{e}+2ce_{R}\cdot\omega_{e}\\
V_{x}g_{1} & =\left(a\omega_{e}+be_{R}\right)^{\top}\\
V_{x}g_{2} & =\left(a\omega_{e}+be_{R}\right)^{\top}
\end{align*}
Substituting these expressions in the state penalty \eqref{eq:Hinf_InvOpt_StatePenalty}
of Lemma \ref{lem:InverseOptimal_Prelim}, we obtain the following
expression:
\begin{multline*}
l\left(R_{e},\omega_{e}\right)=-4\left[\frac{b}{2}J\omega_{e}\cdot E^{\top}(R_{e})\omega_{e}+2ce_{R}\cdot\omega_{e}\right.\\
\left.-\left(\frac{1}{r}-\frac{1}{\gamma^{2}}\right)\left|a\omega_{e}+be_{R}\right|^{2}\right]
\end{multline*}

For the first claim, we note that Lemma \ref{lem:InverseOptimal_Prelim}
and the almost-global asymptotic stability of the auxiliary system
\eqref{eq:AuxiliarySystem}, established in Theorem \ref{thm:Theorem_AuxiliarySystem},
directly imply that the control law \eqref{eq:StateFeedback_Main_A}
minimizes the cost functional \eqref{eq:CostFunctional} with the
above state penalty. Moreover, substituting $c=ab\alpha$ from \eqref{eq:Conditions_AuxSys_A}-\eqref{eq:alpha},
we note that the state penalty can be re-stated as
\[
l=4a^{2}\alpha\left\Vert \omega_{e}\right\Vert ^{2}+4b^{2}\alpha\left\Vert e_{R}\right\Vert ^{2}-2bJ\omega_{e}\cdot E^{\top}\left(R_{e}\right)\omega_{e}.
\]
The second claim also follows from Lemma \ref{lem:InverseOptimal_Prelim},
and stipulates that the $\mathcal{L}_{2}$ gain from the extended
disturbance $d_{e}$ to the tracking errors $\left(e_{R},\omega_{e}\right)$
and the control input $u$ is bounded by $\gamma$.
\end{proof}
\begin{rem}
It is interesting to note that in both the state feedback \eqref{eq:StateFeedback_Main_A}
and the worst-case extended disturbance \eqref{eq:WorstCaseExtendedDisturbance},
the second term is proportional to the configuration error vector
$e_{R}$. Recall that in terms of the axis-angle representation of
the attitude error $R_{e}=R_{d}^{\top}R$, this vector can be expressed
as $e_{R}=v\sin\theta_{e}$, where $v$ denotes the axis of rotation
and $\theta_{e}$ the angle of rotation between the actual and desired
orientations. Thus, we have that the proportional action of the controller
and its disturbance rejection capability increase in the interval
$\left|\theta_{e}\right|\in[0,\pi/2]$, and decrease to zero in the
interval $\left|\theta_{e}\right|\in\left[\pi/2,\pi\right]$.  This
is an inherent limitation of PD controllers (with scalar gains) synthesised
using the chordal metric \eqref{eq:ConfigurationErrorFunction}, and
is known to cause arbitrarily slow convergence for initial errors
arbitrarily close to $180^{\circ}$ \cite{Lee2011ACC}.
\end{rem}

\begin{rem}
The preceding results have been proven for a PD control law \eqref{eq:PD+}
without a feedforward term, but remain valid for a PD+ control law
which uses feedforward compensation \eqref{eq:Feedforward} to cancel
the contribution of the reference-related terms in the error dynamics
\eqref{eq:ErrorDynamics}. This compensation requires accurate knowledge
of the reference trajectory ($\omega_{d}$, $\dot{\omega}_{d}$) and
the inertia properties of the rigid body. In general, this information
is more likely to be available in spacecraft applications. If the
reference or inertia properties are not accurately known, as is common
for small rotorcraft and UAVs, then it might be better to avoid using
the feedforward terms.
\end{rem}

\subsection{Tuning Guidelines }

We recall the inverse optimal $H_{\infty}$ control law from \eqref{eq:StateFeedback_Main_A}:
\[
u=-\frac{2}{r}\left(a\omega_{e}+be_{R}\right),
\]
where the $a$, $b$, and $r$ are subject to the requirements given
by \eqref{eq:Conditions_AuxSys_A}, compactly written as
\begin{equation}
0<b\lambda_{j}<a^{2}\left(\frac{1}{r}-\frac{1}{\gamma^{2}}\right).\label{eq:Tuning_conds}
\end{equation}
The PD parameters ($2b/r$, $2a/r$) can be chosen using any tuning
method and then checked for the condition \eqref{eq:Tuning_conds}
which is quite easy to satisfy. Linearized model of the attitude dynamics
can be especially helpful in fully exploiting the powerful frequency
response methods and especially structured $H_{\infty}$ control design
methods, see \cite{Invernizzi2020GCD} for details. Once the linear
design, having adequate performance and robustness properties, has
been finalized, the condition \eqref{eq:Tuning_conds} can be checked
to guarantee almost-global disturbance rejection properties as stipulated
in Theorem \ref{thm:MainResult}. If the condition \eqref{eq:Tuning_conds}
is not satisfied, one can reduce proportional action $b$ or increase
derivative action $a$. Greater disturbance attenuation can be accomplished
by decreasing $r$, thereby generating stronger control action. This
highlights the trade off between disturbance rejection and control
effectiveness. 

\section{Simulation Results\label{sec:Simulation-Results}}

The results presented in section \ref{sec:Robust-Attitude-Tracking}
may seem too technical to apply in practical problems, however, we
illustrate through an example that nonlinear $H_{\infty}$ guarantees
can be easily obtained for a control law design through standard (linear)
control synthesis techniques. 

We consider the attitude control problem of a small satellite and
demonstrate a simple design approach leading to control laws with
appealing performance. To this end, we consider the problem setup:
\[
J=\text{diag}\left(10,10,8\right)\text{kg\ensuremath{\text{m}^{2}}},
\]
\[
d=\begin{bmatrix}0.005-0.05\sin(2\pi t/400)+\delta(200,0.2)+v_{1}\\
0.005-0.05\sin(2\pi t/400)+\delta(250,0.2)+v_{2}\\
0.005-0.03\sin(2\pi t/400)+\delta(300,0.2)+v_{3}
\end{bmatrix},
\]
where $v_{1}$, $v_{2}$, $v_{3}$ are the zero mean white Gaussian
noises with variances $\sigma^{2}=0.015^{2}$. The desired angular
velocity is:
\[
\omega_{d}=\begin{bmatrix}0.05\sin(2\pi t/400)\\
-0.05\sin(2\pi t/400)\\
0.03\sin(2\pi t/400)
\end{bmatrix},
\]
and the desired attitude kinematics ($\dot{R}_{d}=R_{d}\omega_{d}^{\times}$)
is initialized with $R_{d}(0)=I$. The attitude dynamics is initialized
with $\omega=(0,0,0)$ and rotation matrix equivalent of the quaternion
$q(0)=\begin{bmatrix}0.3 & 0.2 & 0.3 & -0.8832\end{bmatrix}^{\top}$.
For direct comparison with a relevant result based on quaternions,
this problem setup has been chosen as in \cite{Luo2005}, except for
the variance of the white noises in $d$ which is taken higher for
this paper. 

For control synthesis, the decoupled single-axis model is considered
and linearized around the origin. This simplified model, along with
some nominal performance and robustness weights, is used to tune the
(linear) structured $H_{\infty}$ controller via nonsmooth optimization
techniques available in MATLAB. After some iterations of the algorithm,
a linear control law meeting all the requirements on the single-axis
linearized is found: 
\[
u=-\left(k_{P}\theta_{e}+k_{D}\omega_{e}\right),
\]
\begin{align*}
k_{P} & =0.9475, & k_{D}=7.2836.
\end{align*}

One could simply take the values of $k_{P}$, $k_{D}$ and check if
the nonlinear control law
\[
u=-k_{D}\omega_{e}+k_{P}e_{R},
\]
would meet the requirements \eqref{eq:Tuning_conds}, where we have
taken $k_{P}=\frac{2b}{r}e_{R}$, $k_{D}=\frac{2a}{r}\omega_{e}$
. It can be easily verified that the nonlinear controller, with $k_{P},k_{D}$
given above, meets the conditions \eqref{eq:Tuning_conds} and therefore
provides a nonlinear $H_{\infty}$ guarantee, i.e., the energy gain
from extended disturbance to the tracking errors and control input
is upper bounded by a finite constant $\gamma$. 

The performance of the controller on both linearized single-axis model
and full nonlinear $SO(3)$ model can be examined in the filtered-step
response provided in Figure \ref{fig:Filtered-step-reference}. Clearly,
the response of linearized system with a linear controller is not
the same as the fully nonlinear counterpart, however, the difference
is not so great considering that the magnitude of the step input is
2 radians. Therefore, it may be useful to start the control design
based on a linearized model (at least for a first-cut design) and
then use this design to select the gains of the nonlinear system.
\begin{figure}[tbh]
\includegraphics[width=0.9\columnwidth]{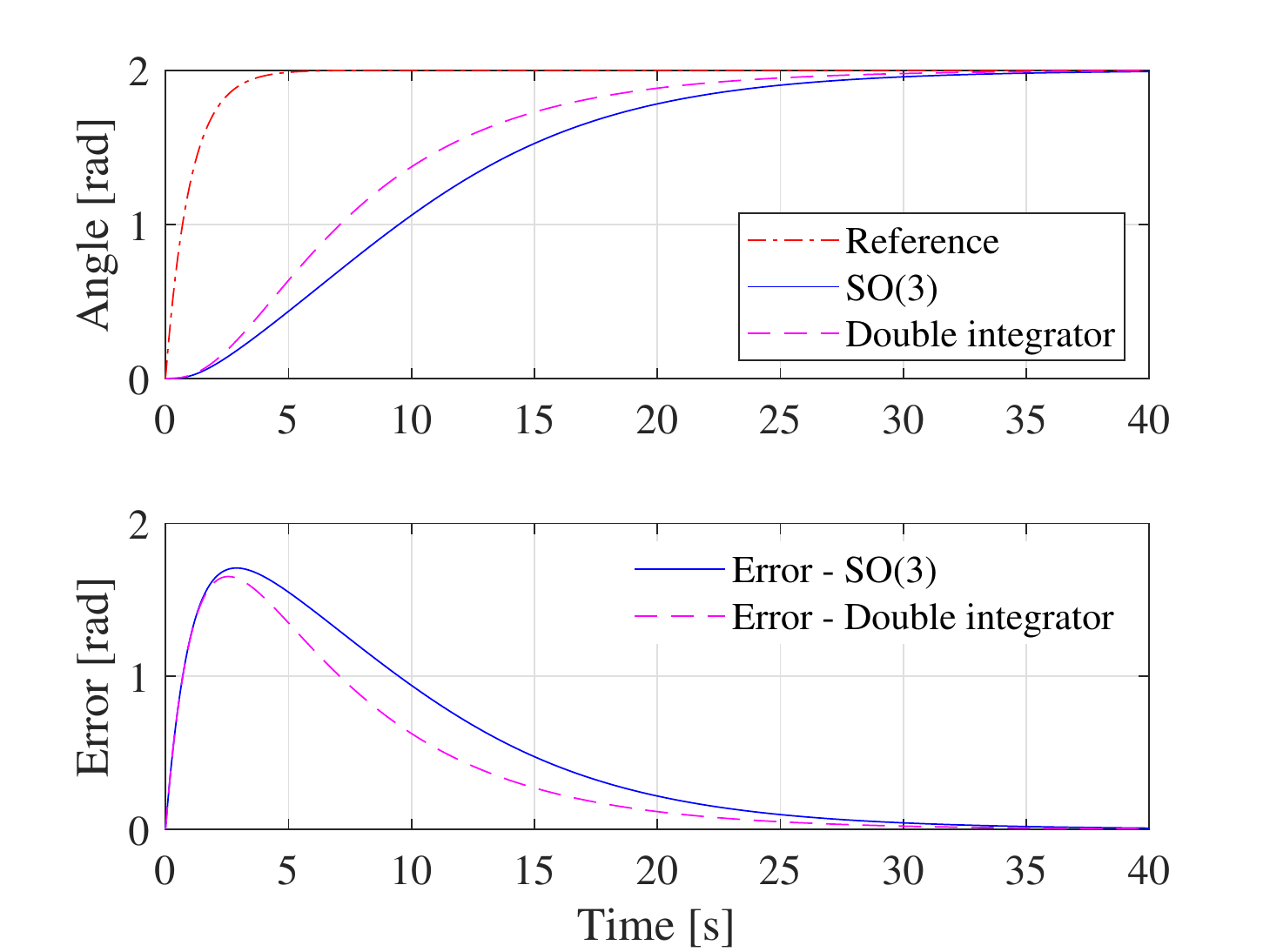}

\caption{The filtered step reference of $2$ radians, and response on $SO(3)$
and the linearized single axis model. \label{fig:Filtered-step-reference}}
\end{figure}

The reference tracking capability of the proposed controller is evaluated
on the problem setup given above. The reference attitude trajectory
along with the achieved attitude is given in Figure \ref{fig:Attitude-tracking}.
The error in attitude is plotted in terms of Euler angles and reported
in Figure \ref{fig:Euler-angular-errors}, with Euler rate errors
in Figure \ref{fig:Rate-errors--}. The components of control torque
vector demanded by the proposed controller are shown in Figure \ref{fig:Input---}. 

The responses show adequate tracking and disturbance rejection properties
and are no worse than those reported in literature for a quaternion
based result \cite{Luo2005}. 
\begin{figure}[tbh]
\includegraphics[width=0.9\columnwidth]{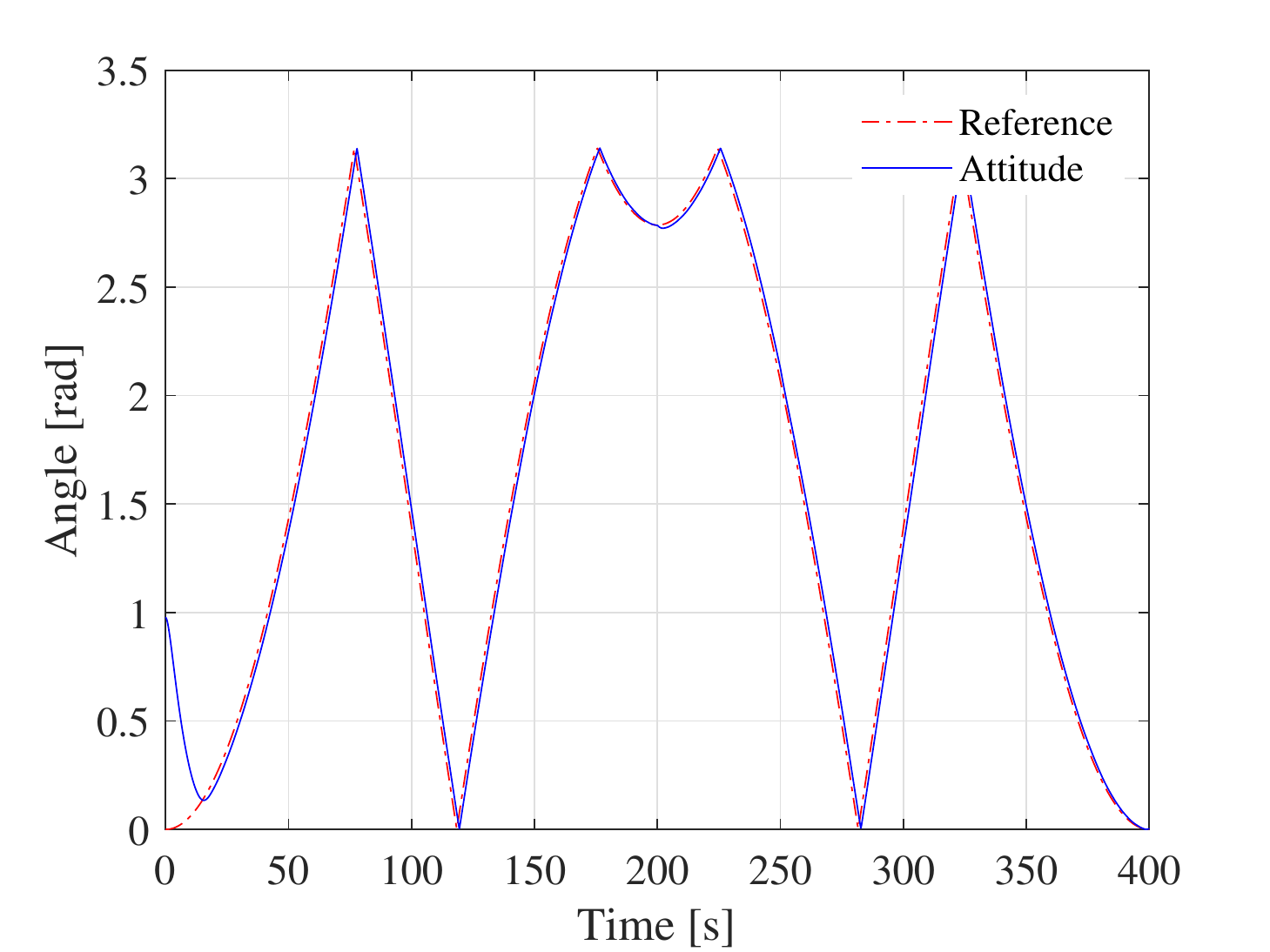}

\caption{The plot of reference trajectory and actual attitude, in terms of
angle or axis angle representation. \label{fig:Attitude-tracking}}
\end{figure}
\begin{figure}[tbh]
\includegraphics[width=0.9\columnwidth]{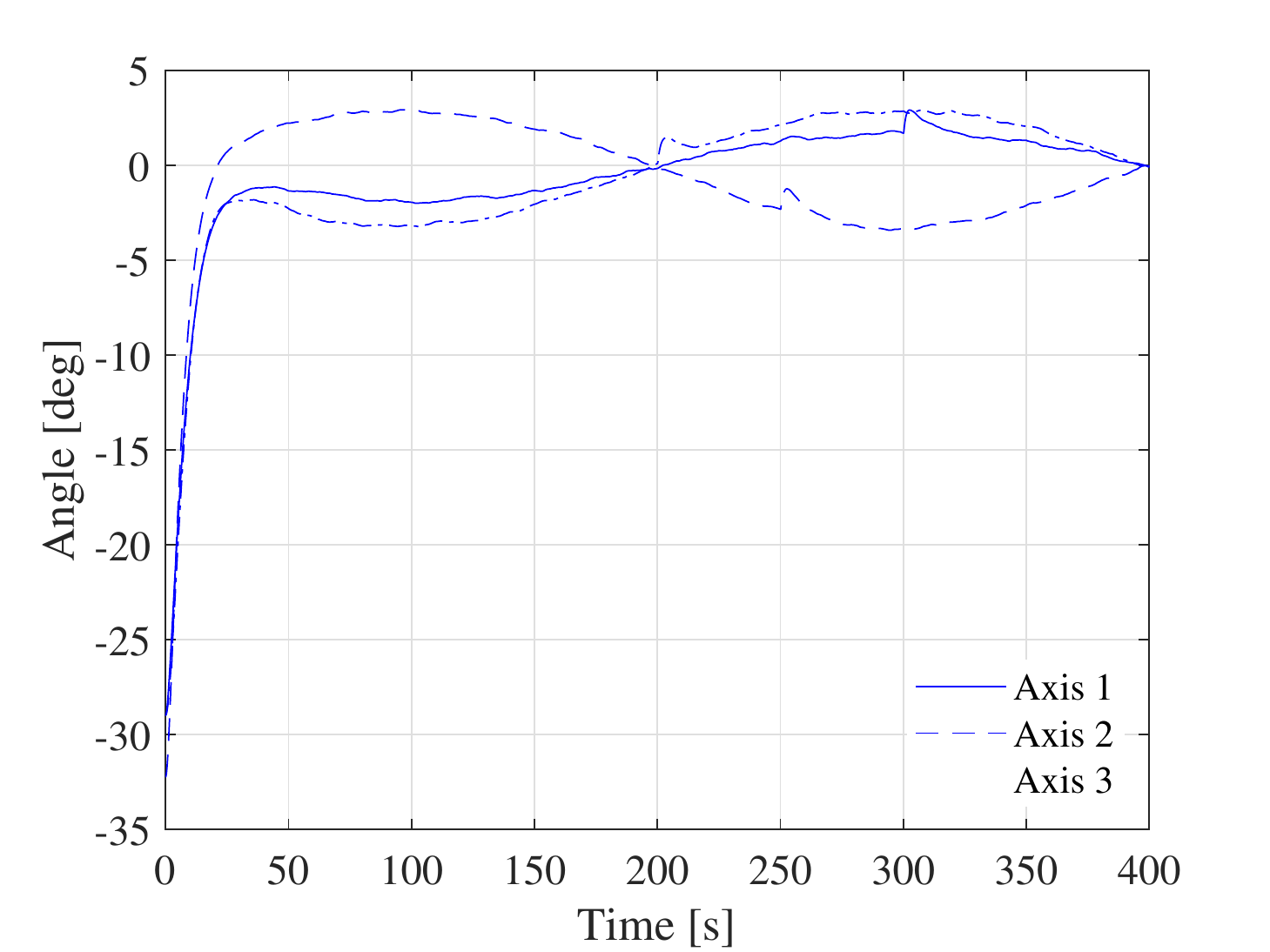}

\caption{The plot of tracking error, in terms of three Euler angular errors.
\label{fig:Euler-angular-errors}}
\end{figure}
\begin{figure}[tbh]
\includegraphics[width=0.9\columnwidth]{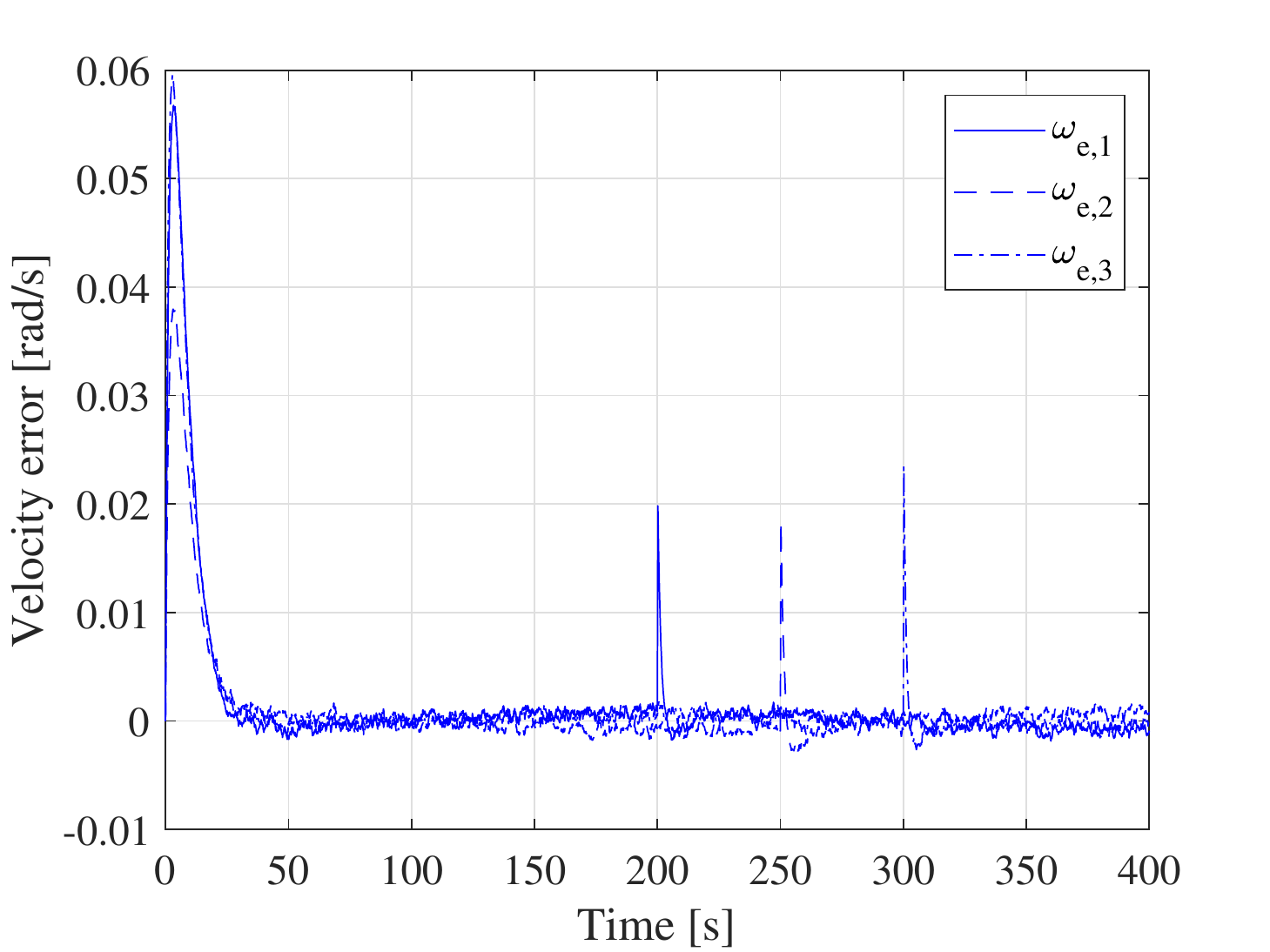}\caption{The plot of angular rate errors, in terms of Euler angle rates. \label{fig:Rate-errors--}}
\end{figure}
\begin{figure}[tbh]
\includegraphics[width=0.9\columnwidth]{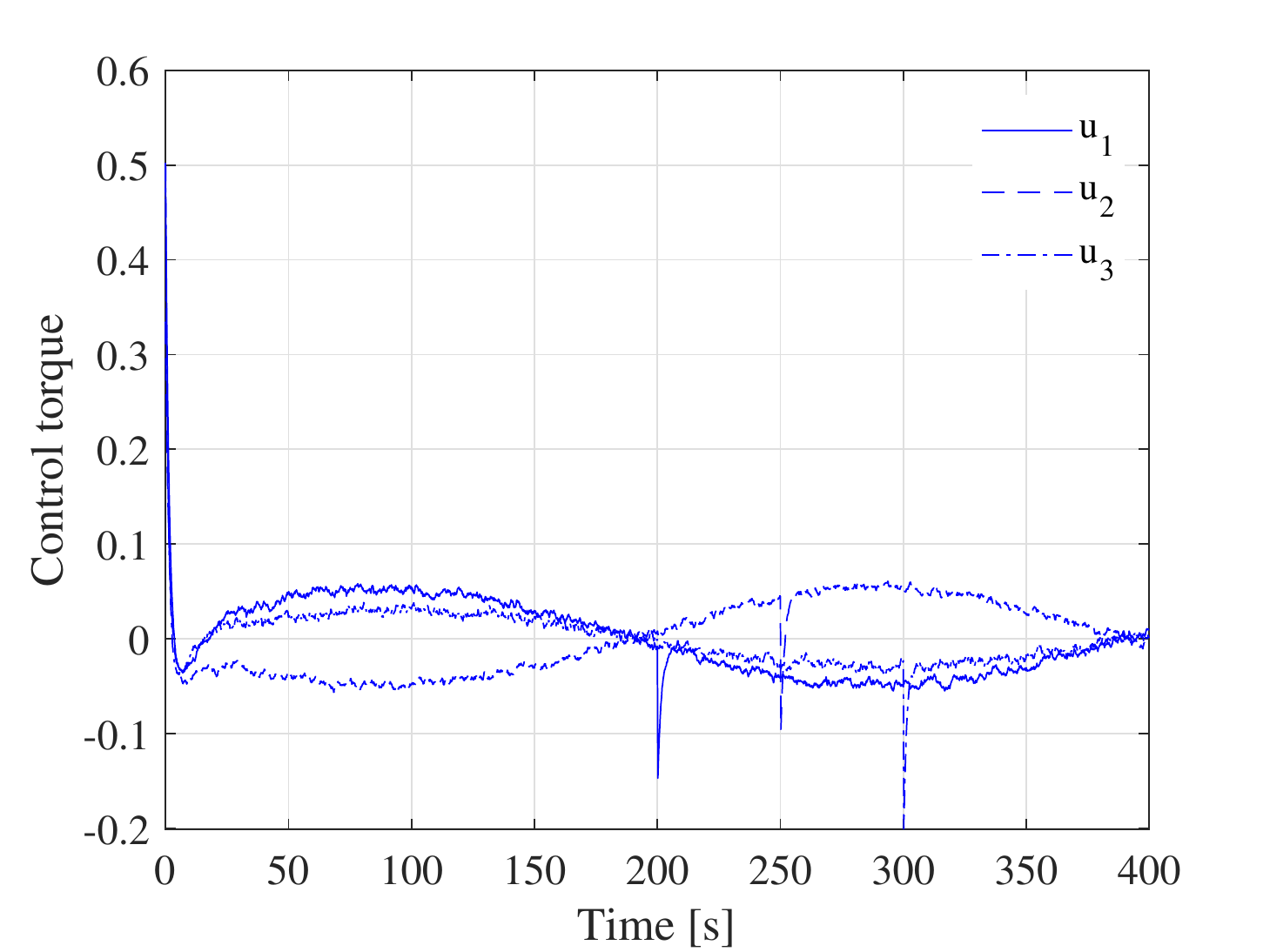}\caption{The plot of components of control torque vector demanded by the proposed
control law. \label{fig:Input---}}
\end{figure}

\section{Conclusion}

The disturbance attenuation problem for the attitude tracking using
rotation matrices has considered and addressed using an inverse optimal
approach. It has been shown that the energy gain from disturbances
to the tracking error is upper bounded by a given constant $\gamma$
if certain mild conditions on the controller gains are satisfied.
A practical method for the tuning of controller gains has been illustrated
using a numerical example, which draws on the powerful structured
$H_{\infty}$ tuning method for the linearized single axis model.
The proposed controller is found to exhibit competitive performance
in the numerical simulations.

\section{Proof of Theorem \ref{thm:Theorem_AuxiliarySystem}\label{sec:Appendix_A}}
\begin{proof}
Replacing the disturbance $d_{e}$ in \eqref{eq:Vdot_Sys1_A} by the
term $\frac{1}{\gamma^{2}}\left(a\omega_{e}+be_{R}\right)$ and substituting
the feedback law \eqref{eq:StateFeedback_A}, the Lyapunov rate for
the auxiliary system can be expressed as

\begin{multline*}
\dot{V}_{\text{aux}}=\frac{b}{2}J\omega_{e}\cdot E^{\top}(R_{e})\omega_{e}-\left(\frac{1}{r}-\frac{1}{\gamma^{2}}\right)\left\Vert a\omega_{e}+be_{R}\right\Vert ^{2}\\
+2ce_{R}\cdot\omega_{e}.
\end{multline*}
Expanding the second term, and using the expression for $\alpha$
in \eqref{eq:alpha}, the Lyapunov rate simplifies to

\begin{multline*}
\dot{V}_{\text{aux}}=-a^{2}\alpha\left\Vert \omega_{e}\right\Vert ^{2}-b^{2}\alpha\left\Vert e_{R}\right\Vert ^{2}-2\left(ab\alpha-c\right)e_{R}\cdot\omega_{e}\\
+\frac{b}{2}J\omega_{e}\cdot E^{\top}(R_{e})\omega_{e}.
\end{multline*}
Using the MATLAB Symbolic Computation Tool, it can be shown that the
matrix 2-norm of $E(R_{e})$ is
\[
\left\Vert E\left(R_{e}\right)\right\Vert =2.
\]
Therefore, choosing $c$ as in \eqref{eq:Conditions_AuxSys_A}, the
Lyapunov rate can be bounded as follows:
\[
\dot{V}_{\text{aux}}\leq-a^{2}\alpha\left\Vert \omega_{e}\right\Vert ^{2}-b^{2}\alpha\left\Vert e_{R}\right\Vert ^{2}+b\lambda_{j}\left\Vert \omega_{e}\right\Vert ^{2}
\]
This can be expressed more compactly as
\begin{equation}
\dot{V}_{\text{aux}}\leq-x^{\top}M_{1}x,\label{eq:Vdot_Aux_UB}
\end{equation}
where $x=\left[\left\Vert e_{R}\right\Vert ,\left\Vert \omega_{e}\right\Vert \right]^{\top}$,
and $M_{1}\in\mathbb{R}^{2\times2}$ is given by
\begin{equation}
M_{1}:=\begin{bmatrix}b^{2}\alpha & 0\\
0 & a^{2}\alpha-b\lambda_{j}
\end{bmatrix}.\label{eq:M1}
\end{equation}
Therefore, the conditions on $a$, $b$ and $\alpha$, given in \eqref{eq:Conditions_AuxSys_A}-\eqref{eq:alpha},
ensure that $M_{1}$ is positive definite. In addition, these conditions
also ensure that
\[
2acI=2a^{2}b\alpha I>2b^{2}\lambda_{j}I>b^{2}J.
\]
Consequently, the sufficient condition \eqref{eq:V_PosDef} for the
positive-definiteness of $V$ is satisfied. This shows that the desired
equilibrium of the auxiliary system \eqref{eq:AuxiliarySystem} is
asymptotically stable, and that $e_{R},\omega_{e}\rightarrow0$ as
$t\rightarrow\infty$.

Next, we show exponential stability. Define
\[
U=\frac{a}{2}\omega_{e}\cdot J\omega_{e}+ab\alpha\Psi.
\]
From \eqref{eq:ConfigurationErrorFunction_TimeDerivative}, \eqref{eq:AuxiliarySystem},
\eqref{eq:StateFeedback_A}, and \eqref{eq:alpha}, it follows that:
\begin{align*}
\dot{U} & =a\omega_{e}\cdot J\dot{\omega}_{e}+ab\alpha e_{R}\cdot\omega_{e}\\
 & =a\omega_{e}\cdot\left(-\omega_{e}^{\times}J\omega_{e}-\alpha\left(a\omega_{e}+be_{R}\right)\right)+ab\alpha e_{R}\cdot\omega_{e}\\
 & =-a^{2}\alpha\left\Vert \omega_{e}\right\Vert ^{2}
\end{align*}
This implies that $U\left(t\right)$ is non-increasing. Therefore,
for the set of initial conditions in \eqref{eq:InitialConditions_ExponentialStability},
we obtain
\begin{align*}
\Psi\left(R_{e}\left(t\right)\right)\leq\frac{1}{ab\alpha}U\left(t\right) & \leq\frac{1}{ab\alpha}U\left(0\right)\\
 & \leq\frac{\lambda_{j}}{2b\alpha}\left\Vert \omega_{e}\left(0\right)\right\Vert ^{2}+\Psi\left(R_{e}\left(0\right)\right)\\
 & \leq\psi.
\end{align*}
Thus, the upper bound in \eqref{eq:Psi_Bounds} is satisfied. Consequently,
from \eqref{eq:Psi_Bounds} and \eqref{eq:CandidateStorageFunction_SO3},
we have that
\begin{equation}
x^{\top}M_{2}x\leq V_{\text{aux}}\leq x^{\top}M_{3}x,\label{eq:V_Aux_Bounds}
\end{equation}
where $x=\left[e_{R}^{\top},\omega_{e}^{\top}\right]^{\top}$, and
$M_{2},M_{3}\in\mathbb{R}^{6\times6}$ are given by
\begin{align*}
M_{2} & =\begin{bmatrix}2cI & bJ\\
bJ & aJ
\end{bmatrix}, & M_{3}= & \begin{bmatrix}\frac{4c}{2-\psi}I & bJ\\
bJ & aJ
\end{bmatrix}.
\end{align*}
We have already seen that the conditions on $\alpha$, $a$, $b$,
and $c$ in \eqref{eq:Conditions_AuxSys_A}-\eqref{eq:alpha} ensure
that $M_{2}$ is positive definite. Now we note that the same conditions
also ensure that $M_{3}$ is positive definite. In particular, a sufficient
condition for $M_{3}$ to be positive definite is given by
\[
\frac{4ac}{2-\psi}I>b^{2}J.
\]
Since $c=ab\alpha$, and $0<2-\psi<2$, it follows from the conditions
in \eqref{eq:Conditions_AuxSys_A}-\eqref{eq:alpha} that
\[
\frac{4ac}{2-\psi}I=\frac{4a^{2}b\alpha}{2-\psi}I>2a^{2}b\alpha I>2b^{2}\lambda_{j}I>b^{2}J.
\]
Consequently, from \eqref{eq:Vdot_Aux_UB}-\eqref{eq:V_Aux_Bounds},
we conclude that the desired equilibrium is exponentially stable in
$\left(e_{R},\omega_{e}\right)$. Note that in \eqref{eq:InitialConditions_ExponentialStability},
the initial attitude error almost covers $SO(3)$, excluding only the
three attitude errors corresponding to the undesired equilibria which
are unstable. Furthermore, the initial angular velocity errors tend
to cover $\mathbb{R}^{3}$ as $\alpha\rightarrow\infty$. Therefore,
the desired equilibrium is almost semi-globally exponentially stable.
\end{proof}
\bibliographystyle{IEEEtran}
\bibliography{Library}

\end{document}